\documentclass{article}
\usepackage[margin=0.75in]{geometry}
\usepackage{graphicx}
\usepackage{amsmath}
\usepackage{verbatim} 
\usepackage{color}

\author{Sachin Kadloor and Nandakishore Santhi}
\title{Understanding Cascading Failures in Power Grids}

\newcommand{\ml}{\mathcal{L}}
\newcommand{\md}{\mathcal{D}}
\newtheorem{theorem}{Theorem}[section]
\newtheorem{lemma}[theorem]{Lemma}

\newenvironment{proof}[1][Proof]{\begin{trivlist}
\item[\hskip \labelsep {\bfseries #1}]}{\end{trivlist}}

\newcommand{\qed}{\nobreak \ifvmode \relax \else
      \ifdim\lastskip<1.5em \hskip-\lastskip
      \hskip1.5em plus0em minus0.5em \fi \nobreak
      \vrule height0.75em width0.5em depth0.25em\fi}

\usepackage{indentfirst}
\begin{document}
\date{}
\maketitle
\centerline{Draft of \date{\today}}
\begin{abstract}
In the past, we have observed several large blackouts, i.e. loss of power to large areas. It has been noted by several researchers that these large blackouts are a result of a cascade of failures of various components. As a power grid is made up of several thousands or even millions of components (relays, breakers, transformers, etc.), it is quite plausible that a few of these components do not perform their function as desired. Their failure/misbehavior puts additional burden on the working components causing them to misbehave, and thus leading to a cascade of failures. 

The complexity of the entire power grid makes it difficult to model each and every individual component and study the stability of the entire system. For this reason, it is often the case that abstract models of the working of the power grid are constructed and then analyzed. These models need to be computationally tractable while serving as a reasonable model for the entire system. In this work, we construct one such model for the power grid, and analyze it.
\end{abstract}

\section{Introduction}
We model the power grid as a network of nodes with edges. The edges represent generators serving a load; an edge between two nodes represents a logical connection between the nodes, representing an agreement between the nodes to share each others' loads in case of one of the nodes failing. The event of a random number of components of the power grid failing is modeled as a random disturbance causing some of the nodes in the network to fail. The loads at these nodes is shared by their neighbors. It could so happen that this additional burden causes some of these neighboring nodes to fail, causing an additional burden on the remaining nodes, and so on. We study the robustness of such system, i.e., what disturbance levels the system can accept before a failing node (or a few failing nodes) would result in the failure of all the components in the system resulting in a large blackout.

\section{System Model}\label{sec:sysmodel}
Our system model consists of a graph $G(V,E)$, with a node $v\in V$ representing a generator. Associated with each node $v$ is a number $L_v(n)$ representing the load demand at that node at time $n$. Also, associated with each node is a capacity $L_{max}$. A node is said to be alive as long as the load demand at that node is within its capacity, i.e., as long as $L_v(n)<L_{max}$ that node is alive. If the load demand at a node exceeds its capacity, then it fails. The edges on the network, $E$, represent a link between the generators with the following interpretation: If a node fails, then the load served by that node is equally shared by all the nodes connected to it. Consider the 5 nodes, nodes A through E, in Figure \ref{fig:network}. The nodes represent generators located at various distant locations. The yellow regions around each of the nodes represent the area of coverage of these generators. Depending on the area of coverage, the `steady state' load at different generators could be different. Each generator is assumed to generate sufficient power to feed all the customers' demands. 
\begin{figure}[ht]
\centering
\includegraphics[width=0.6\columnwidth]{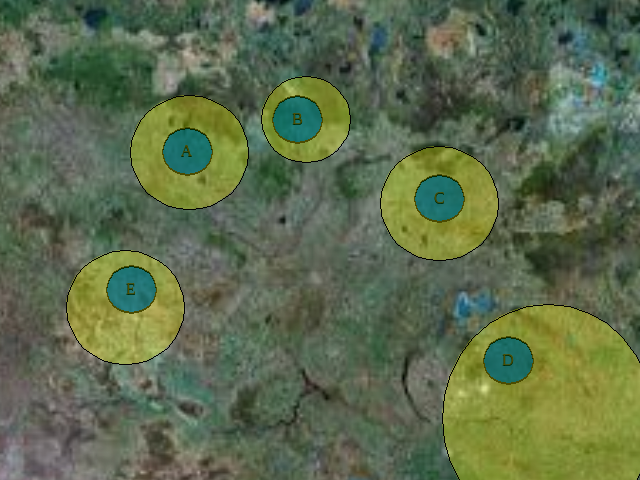}
\caption{Generators, in blue, located at different geographical locations, with yellow regions representing the area of coverage of each of these generators.}
\label{fig:network}
\end{figure}

However, occasionally, a few of these generators fail. This could either be because of failure of the generator itself, or because the load increases to a level greater than the generation capacity of the generator, or because of a natural calamity taking down the power lines which connect the generator to the power grid. Any of these events would result in a blackout to all the customers being served by the failing generator. To bolster the power grid, one could conceive of agreements between different nodes (probably belonging to different utility companies). Links are formed between nodes with the interpretation that the load at a failing node is equally distributed among its neighbors.

We start out by considering the case when every node in the network is assumed to have the same maximum capacity $L_{max}$ of one unit. The load at each node is assumed to be a constant number $a_0\in(0,1)$. Hence, $L_v(0)\sim \delta(l-a_0), \forall v\in V$, i.e., if $L_v(0)$ were to be thought of as a random variable, it would have a distribution $\delta(l-a_0)$. The assumptions about the initial distribution are not required for our analysis, but the resulting equations are easier to analyze. We will consider other initial distributions later. 

At time $0$, a disturbance $D_v(0)$ is added to node $v$. As discussed before, a variety of events could trigger the chain of cascading failures. We model this initial trigger as an abrupt increase in the consumers' load demand. The disturbance $D_v(0)$ is modeled as a random variable with distribution $Exponential(d_m)$, with a mean value of $d_m<1$. The disturbances at various nodes are independent of each other. As a result of this disturbance, the load at node $v$ is now $L_v(0)+D_v(0)$. At time $1$, all those nodes for which $L_v(0)+D_v(0)>1$ fail. All the neighboring alive nodes share the load of the failing nodes. If node $i$ fails, and suppose it is connected to nodes $r$ and $s$, which have not yet failed,  then each of these two nodes share the load of node $i$, that is, the load at node $r$ at time $1$ is $L_r(1)=L_r(0)+D_r(0)+\frac{L_i(0)+D_i(0)}{2}$ and the load at node $s$ at time $1$ is $L_s(1)=L_s(0)+D_s(0)+\frac{L_i(0)+D_i(0)}{2}$. Now, if $L_r(1)=L_r(0)+D_r(0)+\frac{L_i(0)+D_i(0)}{2}>1$, then node $r$ fails, and all of its load is shared among the nodes connected to it, and so on. Note that if no nodes fail at time $n$, no other nodes will fail subsequently.



\paragraph*{Large Fully Connected Networks:}
We now limit our attention to large fully connected networks. The fully connected assumption means that every node in the network has agreed to share the load of every other node in the event of a failure. From a robustness point of view, intuitively, fully connected graphs should be the most resilient to disturbances. This intuition was confirmed partly by the following simulation. Refer Appendix \ref{app:simdesc} for details on the simulation. We considered random graphs with varying probability of the presence of an edge. Random graphs are graphs in which any two nodes are connected with a certain probability. The presence or the absence of an edge between two nodes in the graph is independent of the presence or absence of any other edge. As described in the previous section, we ran the simulation and computed the fraction of the simulations in which there was no outage\footnote{A slight deviation from the model described earlier is that for this simulation we assumed that the initial distribution of loads is $Uniform[0,1]$ instead of $\delta(l-a)$.}. The results are shown in Figure \ref{fig:p_outage}. It is evident from the plot that when there are more number of connections between the nodes, there is a smaller probability that there is any outage in the region. Also, when the total number of nodes in the network is large, the improvements in going from weakly connected to fully connected are higher. We hence assume that most networks designed should be fully connected. Recall that by fully connected, we do not necessarily mean that there is a physical wired connection between any two generators, but only an agreement between the nodes that they would share each others' loads in the event of a failure. 

\begin{figure}[ht]
\centering
\includegraphics[width=0.6\columnwidth]{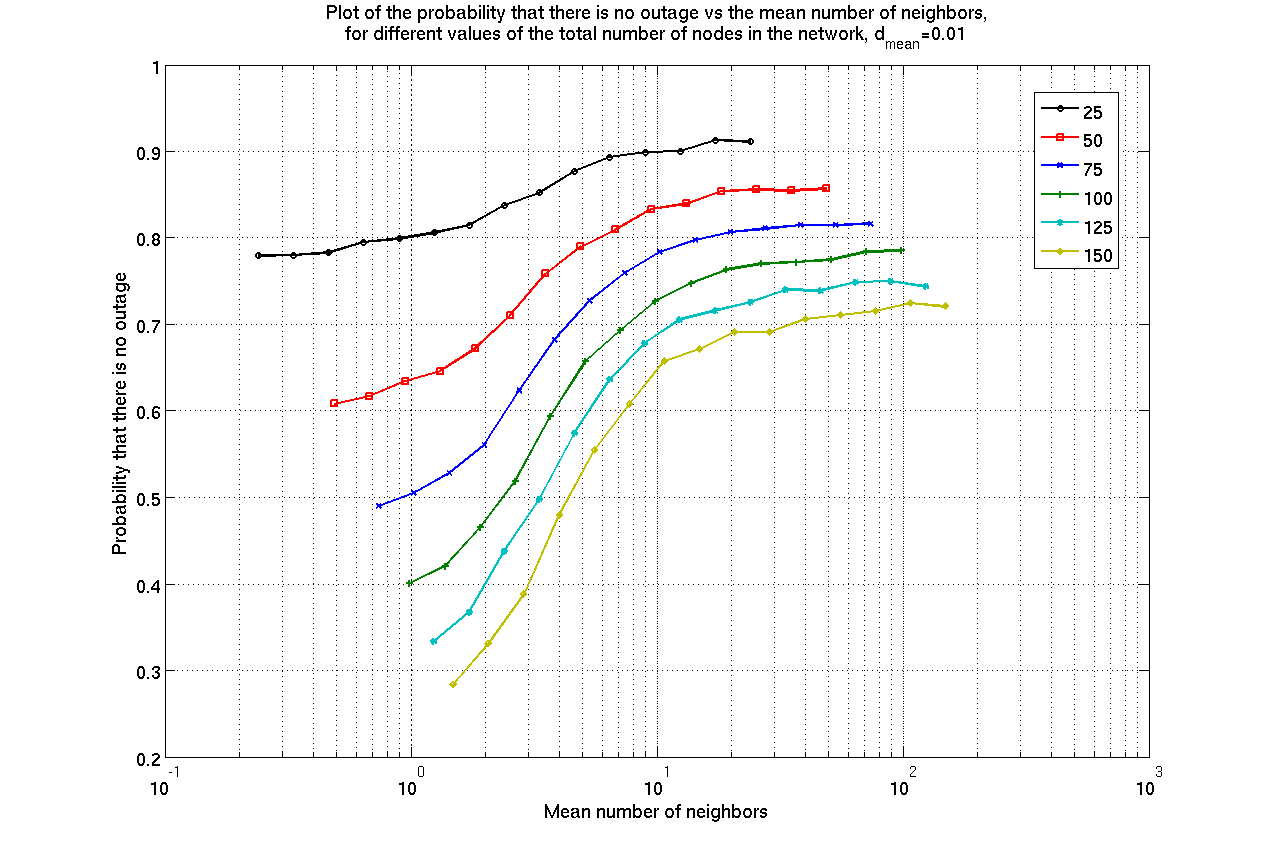}
\caption{The probability that there is no outage increases with the connectivity in the network.}
\label{fig:p_outage}
\end{figure}

Note that we say the results in Figure \ref{fig:p_outage} only partially confirm that fully connected graphs are the best in terms of providing resilience to disturbances because, the network resilience could be measured using different metrics. The metric we used to make that claim is the probability that there are no outages. We could also measure this resilience using another metric, on an average, the fraction of the population which goes into an outage as a result of the disturbance. Figure \ref{fig:fraction_outage} shows the fraction of population which goes into an outage in the event of a disturbance. The reason for this behavior is the following. The fully connectedness of the graph ensures the following: Every node is an immediate neighbor of every other node. Hence, if there are nodes which have survived the disturbance, then they will supply power to all the regions being served by the failed nodes, resulting in no outage in any of the regions. However, if all the nodes end up failing, then the entire region goes into an outage. \emph{When the graph is fully connected, either the entire region survives the disturbance, or all of it goes into an outage}. If the graph were not fully connected, when a few of the nodes fail, it could divide up the region into islands of disconnected nodes which could individually survive or fail. 

This behavior is seen in Figure \ref{fig:distribution}, where we plot the results of individual experiments. When there are few edges in the graph, there are many `small' outages, when the edges are many, there are few, but `large' outages. It is hence not very clear if fully connectedness of the graph is indeed the best to have. From a customer's point of view, large graphs are indeed the best as they have the least probability of an outage. A customer is not bothered if there is an outage only in his locality or the entire region. From the utility company's point of view, if the graph is fully connected, there is a higher chance that the entire region goes into an outage. In such an event, the time required by the utility company to fix all the generators (by replacing the defective components, re-setting the circuit breakers, etc.) will be larger compared to fixing only a few nodes. To come up with possibly other metrics of resilience and to decide on what degree of connectivity is best is a study in itself (we have considered only random graphs, perhaps a structured graph with low degree of connectivity could offer best of both extremes), and we explore this question no further. We just concentrate on fully connected graphs.  

By a large graph, we mean that we will study the behavior of network in the limit that the number of nodes goes to infinity. This simplifies the analysis greatly. Also, as we shall see, the behavior of a network with finite, but large number of nodes is well approximated by the behavior of the large network.  

\begin{figure}[ht]
\centering
\includegraphics[width=0.6\columnwidth]{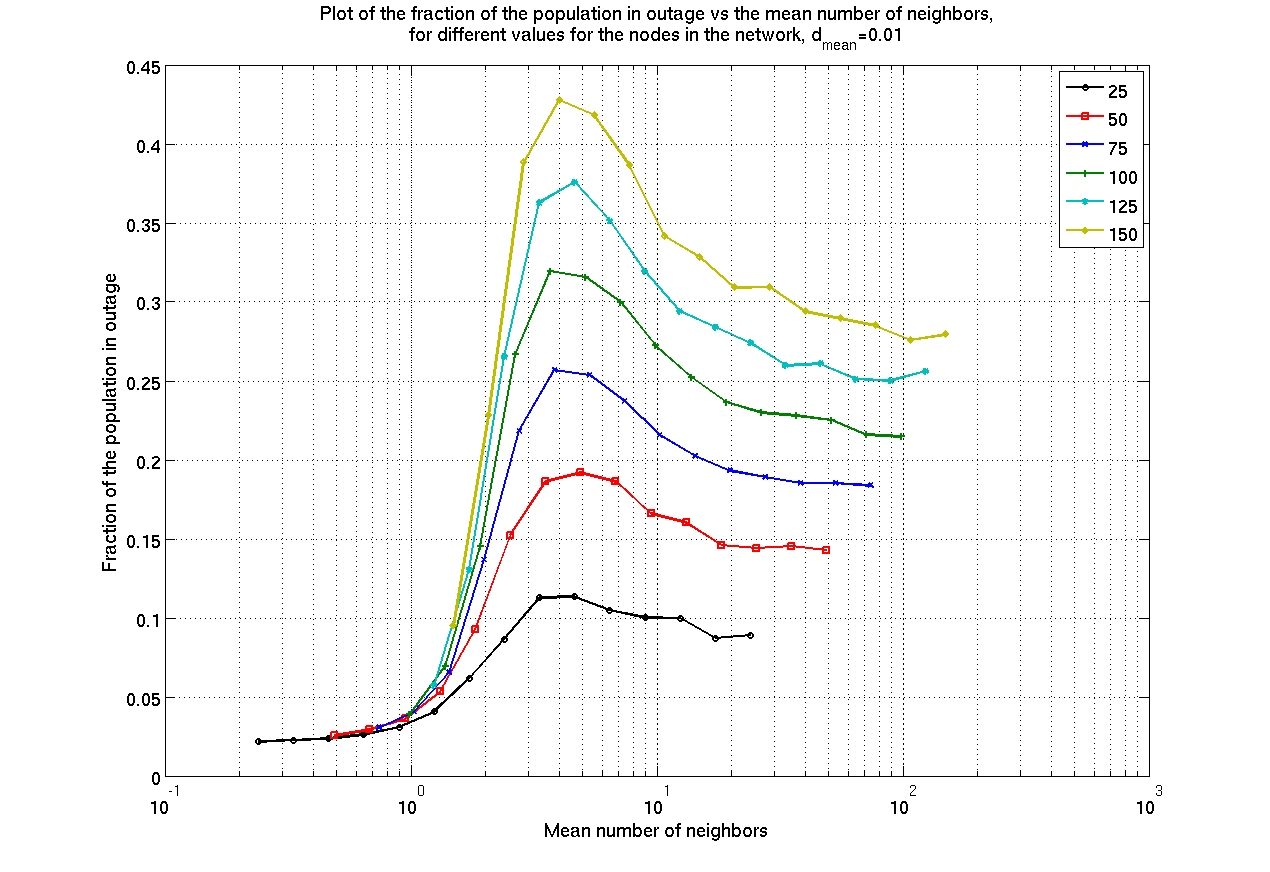}
\caption{The fraction of the population which goes into an outage, averaged over all the simulations.}
\label{fig:fraction_outage}
\end{figure}

\begin{figure}[ht]
\centering
\includegraphics[width=0.6\columnwidth]{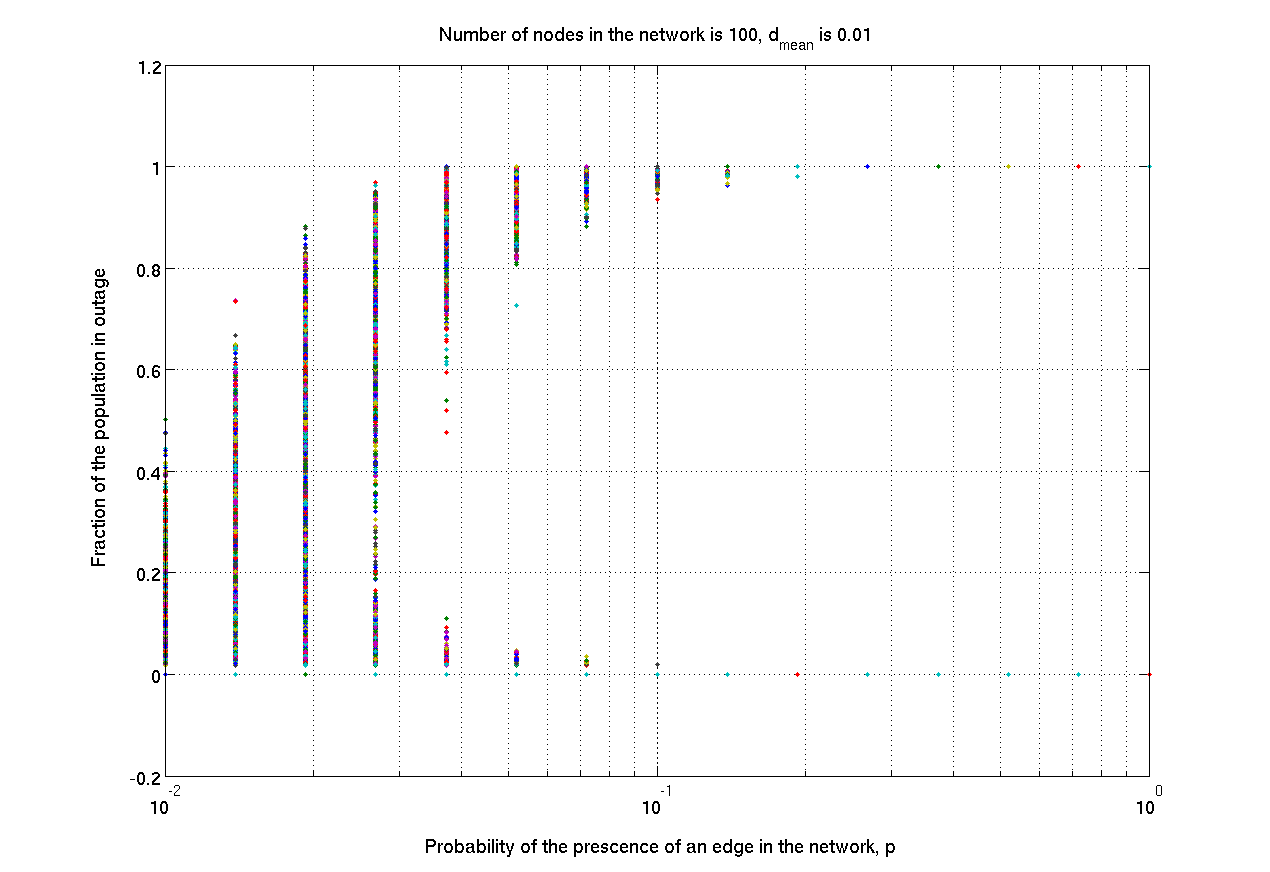}
\caption{Results of individual experiments plotted for the case when there are 50 nodes in the network, and $d_m$=0.1.}
\label{fig:distribution}
\end{figure}

\paragraph*{Main Result of the paper:} 
\emph{For a fully connected large graph, we will show that there is a threshold value $d_{critical}$ with the property that, when $d_m < d_{critical}$, the network survives the disturbance with probability 1, i.e., no region goes in an outage; and when $d_m > d_{critical}$, all the nodes in the network fail with probability 1 leading to a complete outage.} 

\section{Related Works}
Different works have considered different failure models and have analyzed them. Some of these works specifically study cascading failures in power grids whereas others study more general networks, e.g.~social networks, Internet, food chains, etc.. In most of these networks, there is some notion of a failure of a node, a computer connected to others being infected by a virus, people in social networks framing opinions on whom to vote for, etc.. Some of the failure models which are studied are: 
\begin{itemize}
\item \textbf{Each failing nodes increases the load at every other node uniformly: \cite{Dobson04_simple}} Like in the system model described above, there are nodes which form the network, each with a tolerance level of $l_{max}$. However, in this model, the underlying topology of the nodes is ignored. The initial disturbance causes a few of the nodes to fail. When any node fails, the load at every other active node is uniformly increased by a constant load $p$. In this work, the authors study the fraction of the nodes which survive a disturbance, and show that their distribution follows a quasi-binomial distribution.
\item \textbf{Each failing node takes down with it a random number of nodes: \cite{Dobson04_branching}} In this model, the network topology is again ignored. Initially each node is active. Then one of them is assumed to fail. This failing node results in the failure of $M_1$ nodes, where $M_1$ is a random variable distributed according to a p.m.f.~$P_F(\cdot)$. Each of these failed nodes cause a further failure of a certain nodes. At stage $i$, if $M_i$ nodes have failed, then, the total number of failed nodes at stage $M_{i+1}$ is $$ M_{i+1} = M_{i+1}^{1} + M_{i+1}^{2} + \ldots M_{i+1}^{M_i}, $$ where $M_{i+1}^{k}$ is the number of nodes which fail as a result of the failure of the $k^{th}$ failure at time $i$. $M_{i+1}^{k}, k=1,2,\ldots,M_i$ are all assumed i.i.d.~according to the same p.m.f.~$P_F(\cdot)$. The authors study the distribution of $\sum_{k=0}^{\infty} M_k$. The behavior of $M$ is governed by the mean of the number of failures which result from the failure of a node, i.e., $\lambda=\sum_{f=1}^{\infty}fP_F(f).$ When $\lambda<1$, Pr$(M<\infty)=1$. When $\lambda>1$, Pr$(M<\infty)=0$. In the case when $\lambda=1$, the authors show that the total number of failed nodes follows a power law, i.e., Pr$(M=r)\propto r^{-\frac{3}{2}}$. 
\item \textbf{Nodes which fail if a fraction of their neighbors fail : \cite{Watts_random}} In this model, each node is assumed to be in one of the two states, active or failed. Each node fails if a fraction $\phi$ of its neighbors fail. The authors study the behavior of ``random networks''. The network is constructed out of $n$ nodes, where each node has $k$ neighbors, where $k$ is an integer-valued random variable with distribution $P_K$, and mean $z \ll n$. All the nodes in the network are initially in active state. At time $1$, a fraction $\Phi_0 \ll 1$ of the nodes change their state. The system then evolves as discussed above. The authors relate this problem to \emph{percolation theory}, and find the stability region, i.e., the pairs $(z,\phi)$ for which there is no cascading failure. This is done for the case when all the nodes in the system have the same threshold $\phi$. 
\item \textbf{Drop in efficiency of a network because of an imbalance in flow distribution: \cite{Marchiori_physical,Marchiori_efficiency}} In this model, all the edges in the graphs have weights associated with them. The weight represents the cost of using that link. Furthermore, each node also has a weight associated with it, representing its capacity. For example, the nodes could represent routers in a network, the weight of a router represents the maximum number of packets it can process per second. The links represent the topology of the network and the link weights are the costs associated with transmitting a packet on that link. In this model, every node sends packets to every other node in the network at a constant rate of one packet per second. \emph{The packets are routed through the path with the smallest cost. If the path includes multiple links, the cost of the path is the sum of the costs along individual links.} The load $l_i(t)$ on node $i$ at time $t$ is the total number of flows through it, i.e., the total number of most efficient paths through node $i$. The capacity of node $i$ is $C_i=\alpha l_i(0), \alpha \geq 1$. The system is stable at time 0 because the total flow through each node is less than its capacity. At time 1, one randomly chosen node is removed from the network. All the flows through that node are now re-routed. This is done by re-calculating the most efficient paths for these flows. This might lead to an overload at a few of the nodes. When this happens, the weights of all the edges emerging from an overloaded node are increased, i.e., the cost of routing a packet through an overloaded node is increased. When the costs associated with the links change, the most-efficient paths are re-calculated. This re-routing will cause overload at other nodes, leading to an increase in the cost of routing packets through these nodes, followed by a re-calculation of the most efficient path, and so on. The efficiency of the network at each time is the total cost incurred in routing all the flows. Note that in this model, only a single node fails at time 0. After that, no further nodes fail, however, the efficiency of the network keeps decreasing because of  increasing link costs. A cascading failure occurs when there is a substantial drop in the efficiency of the network. The authors study two graphs, \emph{Erdos-Renyi random graphs}, and \emph{scale-free networks}, and show in both the cases that a single node failing can lead to a cascading loss in the efficiency of the network when $\alpha$, the excess capacity of a node, is close to one.

In a similar work, \cite{Farina_optimization}, the authors consider source nodes, destination nodes, and routing nodes. The links between them represent power lines with associated capacities. The source nodes have a weight representing maximum generation capacity. The destination nodes have a weight representing their power demand. Routing of power from the source nodes to the destination nodes is done so as to meet the constraints on maximum power that can be transmitted along each power line, maximum generation capacity of the generators, and the demand at the loads. Then, the demand of each load is perturbed. Now, the paths are re-calculated, and overloaded lines are cut off. This leads to a re-calculation of the best path. This process is continued iteratively. If all the edges leading to a destination node are cut off, then that node is in outage. The authors use simulation to study how many nodes go in outage.

\end{itemize}

\section{Analysis of failures in the network}
For the case of large fully connected graphs, we will introduce some additional notation. We will define $\mathcal{L}_0$ to be a random variable whose distribution is the same as the distribution of a node which has not failed in stage 0, i.e. $\mathcal{L}_0 \sim \delta(l-a_0)$. More generally, $\mathcal{L}_n$ is a random variable with a distribution same as the distribution of the load at all the nodes which are alive at stage $n$. We will show later that the load at all surviving nodes at any stage is identically distributed. Also, define $\mathcal{D}_n, n\geq1$ to be the `additional disturbance' at stage $n\geq1$. We will make this definition more concrete later on. Define $\mathcal{D}_0$ to be a random variable with distribution $Exponential(d_m)$. 

The initial load at each node at stage 0 is assumed to have the same distribution, of that of $\ml_0$, and the added disturbance at each node is also independent and identically distributed according with the distribution of $\md_0$. Hence, the distribution of the initial load plus the disturbance at each node is also independent and identically distributed with the distribution of $\ml_0+\md_0$. Denote by $p_0$, probability that a node fails at the end of stage 0, i.e. $p_0=Pr(\ml_0+\md_0>1)$. Denote by $N_0$, the number of nodes which fail at stage 0. $N_0$ is a random number: $N_0\sim Binomial(N,p_0)$, where $N$ is the total number of nodes (we are interested in studying the behavior of the system when $N$ is large). Let $k_1,k_2,\ldots,k_{N_0}$ be the indices of nodes which have failed in stage 0. Let $j_1,j_2,\ldots,j_{N-N_0}$ be the indices of the nodes which have not failed at stage 0. Then, 
$$L_{j_i}(1) = L_{j_i}(0)+\frac{1}{N-N_0}(L_{k_1}(0)+L_{k_2}(0)+\ldots+L_{k_{N_0}}(0)), i\in\{1,2,\ldots,N-N_0\}. $$

Knowing the distributions of all the various random variables involved in the above expression, we can find the distribution of $L_{j_i}(1)$. The distributions of the various random variables involved in the above expression can be written as follows:
\begin{enumerate}
\item $L_{j_i}(0) : f_{L_{j_i}(0)}(x) = f_{\ml_0+\md_0}(x|\ml_0+\md_0<1)$.
\item $L_{k_i}(0) : f_{L_{k_i}(0)}(x) = f_{\ml_0+\md_0}(x|\ml_0+\md_0\geq 1)$.
\item $N_0 : N_0\sim Binomial(N,p_0)$.
\end{enumerate}  

Define $\mu_0 \doteq \mathbf{E}(L_{k_i}(0)) = \int_{x=1}^{\infty} x f_{\ml_0+\md_0}(x|\ml_0+\md_0\geq 1) dx  $. We can then express the distribution of $S_{N_0}=\frac{1}{N-N_0}(L_{k_1}(0)+L_{k_2}(0)+\ldots+L_{k_{N_0}}(0))$ as follows. We are interested in the distribution in the limit that the number of nodes, $N$ goes to infinity. 
$$\lim_{N\to \infty} S_{N_0} = \lim_{N \to \infty} \frac{N_0/N}{1-N_0/N}\frac{L_{k_1}(0)+L_{k_2}(0)+\ldots+L_{k_{N_0}}(0)}{N_0}. $$

\begin{lemma}\label{lem:slln}
If $\mu_0<\infty$, $\lim_{N\to \infty} S_{N_0} = \frac{p_0}{1-p_0} \mu_0$
\end{lemma}
\begin{proof}
 Refer Appendix \ref{app:prooflem} for a rigorous proof. Informally, when $N$ is large, the strong law of large numbers dictates that $\lim_{N\to\infty} N_0/N=p_0$. This is because $N_0$ can be thought to be a sum of i.i.d.~Bernoulli random variables with mean $p_0$, and the empirical mean of a large number of i.i.d.~random variables behaves like the statistical mean. Next, when $N$ is large, $N_0$ can also be shown to be very large, in which case, the second term in the expression also behaves like the statistical mean of the random variables involved. \qed
\end{proof}

The fact that $S_{N_0}$ converges to a constant number, and not a random variable is important. This is because the set of random variables $\{L_{j_i}(1)\}$ continue to be independent random variables, which would not be the case otherwise\footnote{If variables $X_1$ and $X_2$ are independent, then $X_1+S$ and $X_2+S$ are independent if $S$ is a constant, but are not if $S$ is a random variable.}. This simplifies the analysis greatly. Moreover, the set of random variables $\{L_{j_i}(1)\}$ are all identically distributed. 

Define $\md_1=\frac{p_0}{1-p_0}\mu_0$. Define $\ml_1$ to be a random variable with the distribution of all the nodes which are alive at the end of stage 0, i.e., $\ml_1\sim f_{L_{j_i}(0)}(l|L_{j_i}(0)<1)$. The nodes which die at the start of stage 1 are those nodes for which $\ml_1+\md_1>1$. Define $p_1\doteq Pr(\ml_1+\md_1>1)$. Because we started out with a large number of nodes, the number of nodes which are alive at the start of stage 1 is also large. Similar to the case before, a few nodes die at the end of stage 1, and their load is redistributed among those which are alive. Define $\mu_1\doteq \int_{x=1}^{\infty} x f_{\ml_1+\md_1}(x|\ml_1+\md_1\geq 1) dx$. It can be shown, by similar arguments as before, that the total re-distributed load is equal to $\md_2\doteq \frac{p_1}{1-p_1}\mu_1$.  

We can continue this analysis to find the distributions of $\ml_2,\ml_3, \ml_4\ldots$ and $\md_2,\md_3,\md_4\ldots$, along with $p_2,p_3,p_4,\ldots$. At each stage, $p_n$ denotes the probability that a node which is alive at stage $n$ will die.\emph{ We say that the cascade of failures stops if $\lim_{n\to\infty}p_n=0$, else we will say that the cascade continues indefinitely and will result in all the nodes failing. }

In the remaining portion of this section, we will compute these distributions and values for our specific choice of the initial distribution of loads and distribution of the initial disturbance. Restating our choices, we will assume that $\ml_0\sim \delta(l-a_0)$, and $\md_0\sim Exponential(d_m)$.

The distribution of $\ml_0+\md_0$ is:

\begin{eqnarray*}
f_{\ml_0+\md_0}(x)=\left\{\begin{array}{cl}
0 & x<a \\
\frac{1}{d_m}e^{-\frac{x-a_0}{d_{m}}} & a_0\leq x \\
\end{array} \right.
\end{eqnarray*}

$p_0$ can then be calculated as:
\begin{eqnarray*}
p_0 &=& Pr(\ml_0+\md_0\geq1) \\
&=& e^{-\frac{1-a_0}{d_m}}
\end{eqnarray*}

Let $\ml_1$ be the random variable with the same distribution as that of the loads at nodes which have not failed at the end of stage 0. The distribution of $\ml_1$, which is also the distribution of $L_{j_i}(0)$ is:
\begin{eqnarray*}
\ml_1 \sim f_{\ml_0+\md_0}(x|\ml_0+\md_0<1)=\left\{\begin{array}{cl}
0 & x<a_0 \\
\frac{1}{1-e^{-\frac{1-a_0}{d_{m}}}}\frac{1}{d_m}e^{-\frac{x-a_0}{d_{m}}} & a_0\leq x < 1 \\
0 & x\geq1\\
\end{array} \right.
\end{eqnarray*}

The distribution of load at a node which has failed at the end of stage 0, namely the load at node $L_{k_i}(0)$ is: 
\begin{eqnarray*}
f_{\ml_0+\md_0}(x|\ml_0+\md_0\geq1) & = & \frac{1}{d_{m}}e^{-\frac{x-1}{d_{m}}}, \quad  x \geq 1
\end{eqnarray*}
As discussed before, because we are interested in the case when the number of nodes is large, the total re-distributed load as a result of failures is a constant number. The total redistributed load at stage 1 is $\md_1=\frac{p_0}{1-p_0}\mu_0$, where $\mu_0=\mathbf{E}(L_{k_i}(0))=1+d_m$. 

The distribution of all the nodes which have survived the initial distribution is the distribution of the random variable $\ml_1+\md_1$, which is, 
\begin{eqnarray*}
f_{\ml_1+\md_1}(x)=f_{\ml_0+\md_0+\md_1}(x|\ml_0+\md_0<1) & = & \frac{1}{1-p_0}\frac{1}{d_{m}}e^{-\frac{x-a_0-\md_1}{d_{m}}}, \quad  a_0+\md_1\leq x < 1+\md_1.
\end{eqnarray*}

From the equation above, it is seen clearly that a small fraction of the total loads might have a load exceeding their capacity, and hence will fail. Let $\ml_2$ be the random variable with the same distribution as that of the loads at nodes which have not failed at the end of stage 1. The distribution of $\ml_2$ is,
\begin{eqnarray*}
\ml_2 &\sim& f_{\ml_0+\md_0+\md_1}(x|\ml_0+\md_0<1,\ml_0+\md_0+\md_1<1) \\
& = & \frac{1}{1-e^{-\frac{1-a_0-\md_1}{d_m}}}\frac{1}{d_{m}}e^{-\frac{x-a_0-\md_1}{d_{m}}}, \quad  a_0+\md_1\leq x < 1.
\end{eqnarray*}
Compare this to the distribution of $\ml_1$, they are similar. We will shortly develop a set of recursive equations which describe the distributions of all the subsequent random variables.

The probability that a node which has survived the initial disturbance will fail at the end of stage 0 because the re-distributed load caused its load to go beyond its capacity can be calculated as: 
\begin{eqnarray*}
p_1 &=& \int_{x=1}^{x+\md_1}f_{\ml_0+\md_0+\md_1}(x|\ml_0+\md_0<1) dx \\ 
&=& \frac{e^{-\frac{1-a_0}{d_m}}(e^{\frac{\md_1}{d_m}}-1)}{1-e^{-\frac{1-a_0}{d_m}}}.
\end{eqnarray*}

The mean of the loads at the nodes which fail at the end of first stage is, hence, 
\begin{eqnarray*}
\mu_1 &=& \int_{x=1}^{x+\md_1}xf_{\ml_0+\md_0+\md_1}(x|\ml_0+\md_0<1,\ml_0+\md_0+\md_1>1) dx \\ 
&=& 1+d_m-\frac{\md_1}{e^{\frac{\md_1}{d_m}}-1},
\end{eqnarray*}
and the total redistributed load at the end of stage 1 is $\md_2\doteq \frac{p_1}{1-p_1}\mu_1$. 

The evolution of this system can be written in the form of a set of recursive equations:
\begin{enumerate}
\item Initialize $p_0=e^{-\frac{1-a_0}{d_m}}, \md_1= \frac{p_0}{1-p_0}(1+d_m), a_1=a_0, p_1= \frac{e^{-\frac{1-a_1}{d_m}}}{1-e^{-\frac{1-a_1}{d_m}}}(e^{\frac{\md_1}{d_m}}-1)$.
\item For n from 2 to $N_{\textrm{iteration}}$, where $N_{\textrm{iteration}}$ is a large number, do:
 
If $(\md_{n-1}>(1-a_{n-1}))$ and $(a_{n-1}<1)$, break. Else,
	\begin{enumerate}
	\item $a_n=a_{n-1}+\md_{n-1}.$
	\item $\mu_{n-1}=1+d_m - \frac{\md_{n-1}}{e^{\frac{\md_{n-1}}{d_m}}-1}.$
	\item $\md_{n}=\frac{p_{n-1}}{1-p_{n-1}} \mu_{n-1}$.
	\item $p_{n}=  \frac{e^{-\frac{1-a_{n}}{d_m}}}{1-e^{-\frac{1-a_{n}}{d_m}}} (e^{\frac{\md_{n}}{d_m}}-1).$
        \end{enumerate}	
\end{enumerate}

We are interested in $\lim_{n\to\infty}p_n$, or equivalently $\lim_{n\to\infty}a_n$. If $\lim_{n\to\infty}p_n=0$, or equivalently, if $\lim_{n\to\infty}a_n<1$, then we can conclude that the system survives the outage, else, it would mean that the whole system goes into an outage. We wish to solve these recursive equations listed above in order to find the limits of these sequences, however, they are not tractable. We hence resort to simulations to verify the conditions under which the sequences converge. In Figure \ref{fig:exp_dcrit}, we plot the sequence $\{a_n\}$ for different values of $d_m$, for the initial value $a_0=0.8$. It is clearly seen that there is a critical value $d_{critical}=0.048$ below which the sequence always converges to a number less than 1, implying that the cascade has subsided. When $d_m>d_{critical}$, the sequence converges to 1, implying that the load at all the generators is greater than 1, and consequently a complete outage. In Figure \ref{fig:exp_p}, we plot the sequence $\{p_n\}$. As said before, the sequence either converges to 0, when $d_m<d_{critical}$, or to 1 when $d_m>d_{critical}$.

\begin{figure}[ht]
\centering
\includegraphics[width=0.6\columnwidth]{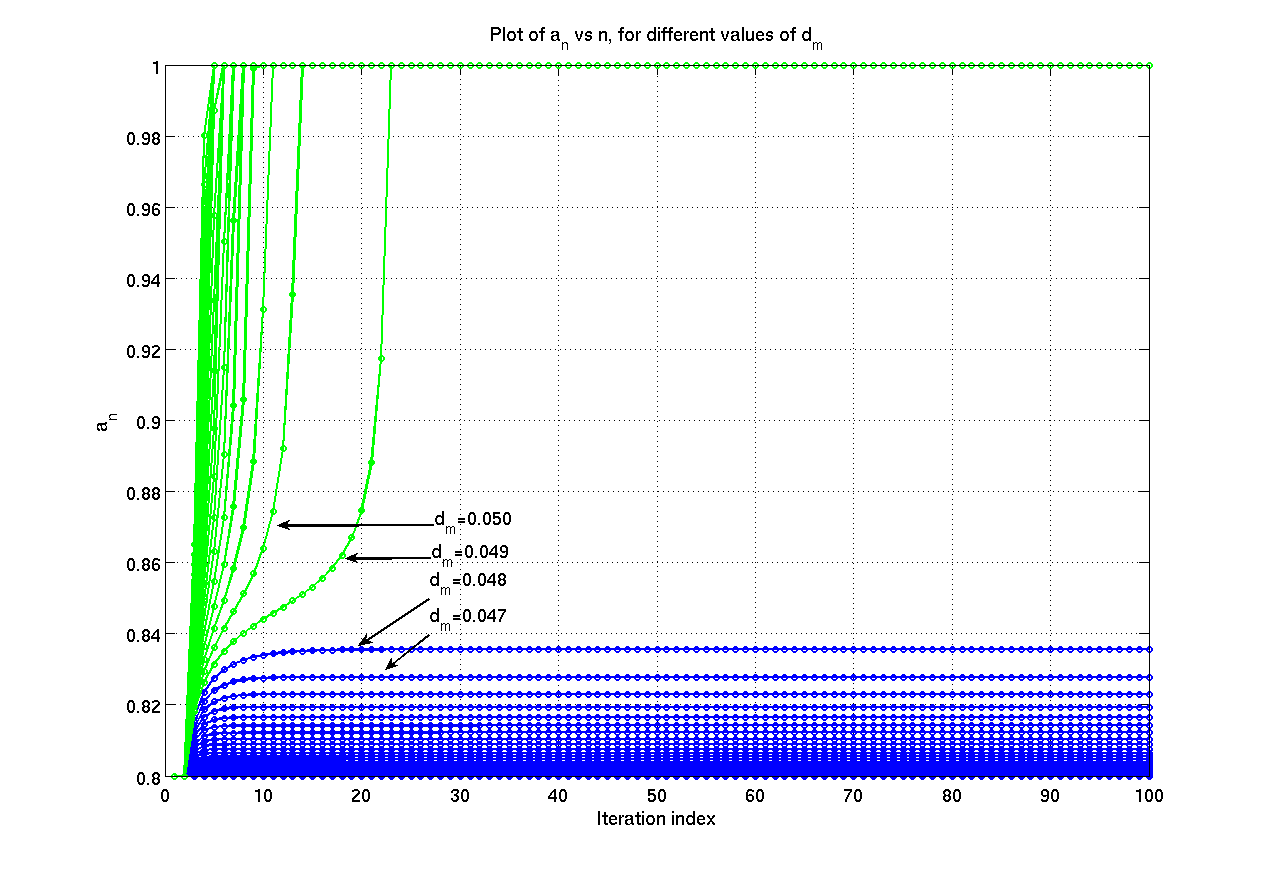}
\caption{Plot of $a_n$ vs $n$ for different values of $d_m$ in the range of 0.001 to 0.07 in increments of 0.001. Note that when $d_m\leq0.048$, $\{a_n\}$s converge to a value less than 1.}
\label{fig:exp_dcrit}
\end{figure}
 
\begin{figure}[ht]
\centering
\includegraphics[width=0.6\columnwidth]{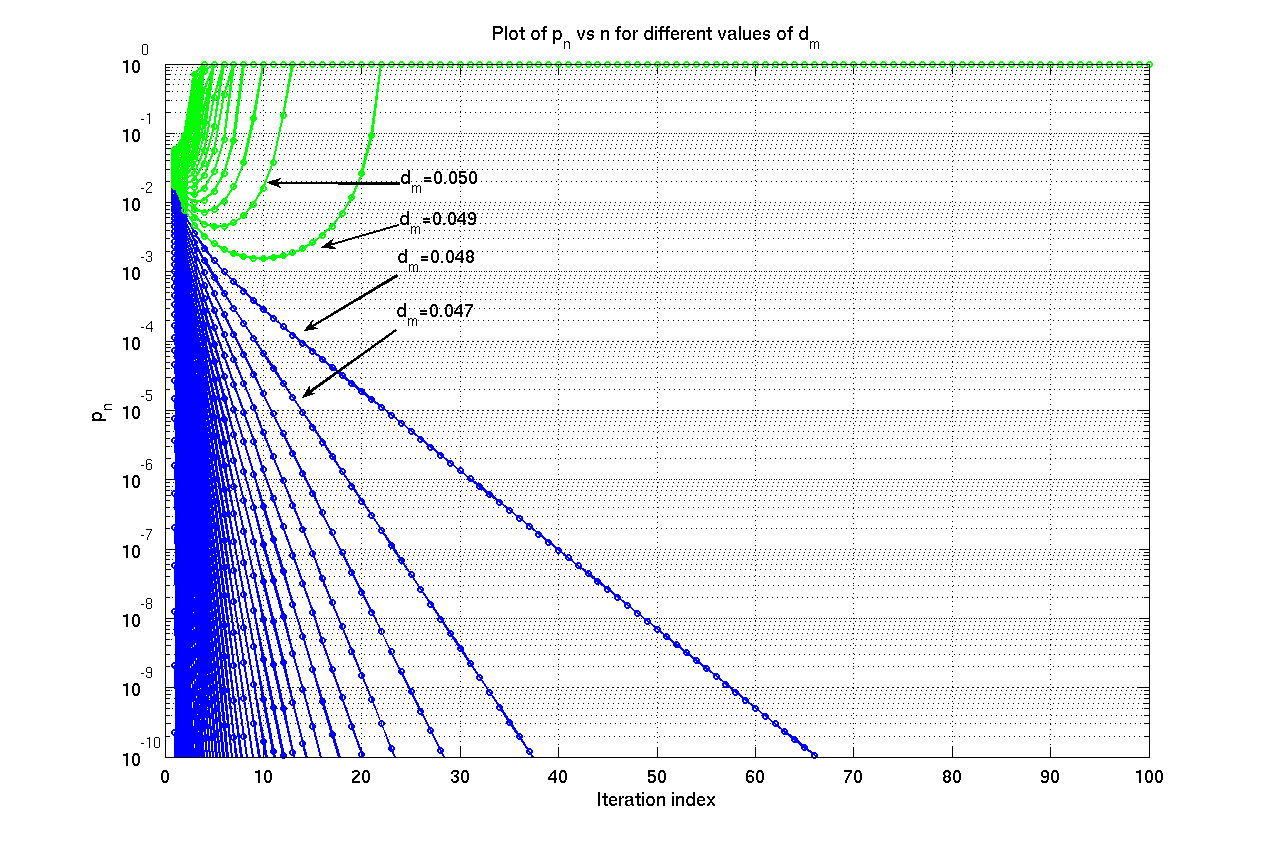}
\caption{Plot of $p_n$ vs $n$ for different values of $d_m$ in the range of 0.001 to 0.07 in increments of 0.001.}
\label{fig:exp_p}
\end{figure}

The theory developed so far can be used to gain further insight into the network resilience of large fully connected graphs. In Figure \ref{fig:exp_d_crit_a}, we plot, in blue, the value of $d_{critical}$ for different values of $a_0$. We also plot, in red, the excess capacity at each node, which is just $1-a_0$. For large values of $a_0$, $d_{critical}$ is about an order smaller than the disturbance the network was provisioned to handle.

\begin{figure}[ht]
\centering
\includegraphics[width=0.6\columnwidth]{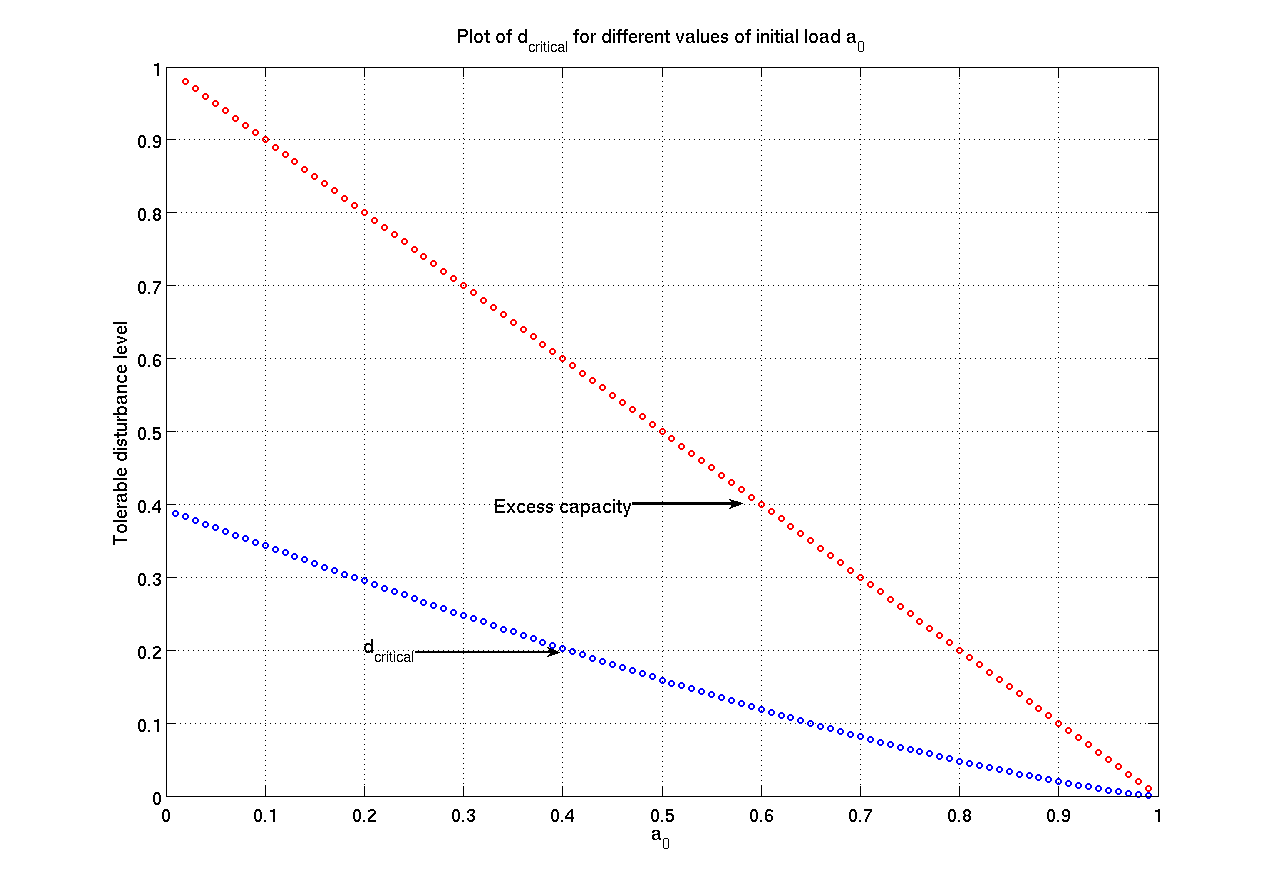}
\caption{Plot of $d_{critical}$ vs $a_0$.}
\label{fig:exp_d_crit_a}
\end{figure}

\section{Further Extensions}
We can perform a similar analysis of systems with different initial load distributions. Suppose that instead of all the nodes in the system having the same initial distribution $a_0$ at time 0, suppose that a fraction $p_a$ of the total nodes had a load of $a_0$, and the remaining fraction, $p_b=1-p_a$, had a load of $b_0$, then the recursive set of equations which describe the system would be:
\begin{enumerate}
\item Initialize $p_0=p_a e^{-\frac{1-a_0}{d_m}}+p_b e^{-\frac{1-b_0}{d_m}}, \md_1=\frac{p_0}{1-p_0}(1+d_m), a_1=a_0, b_1=b_0, p_1=\frac{p_a e^{-\frac{1-a_1}{d_m}}+p_b e^{-\frac{1-b_1}{d_m}}}{1-\left(p_a e^{-\frac{1-a_1}{d_m}}+p_b e^{-\frac{1-b_1}{d_m}}\right)}\left(e^{\frac{\md_1}{d_m}}-1\right).$
\item For n from 2 to $N_{\textrm{iteration}}$, where $N_{\textrm{iteration}}$ is a large number, do:
 
If $(\md_{n-1}<(1-b_{n-1}))$ and $(b_{n-1}<1)$:
	\begin{enumerate}
	\item $a_n=a_{n-1}+\md_{n-1}, b_n=b_{n-1}+\md_{n-1}.$
	\item $\mu_{n-1}=1+d_m - \frac{\md_{n-1}}{e^{\frac{\md_{n-1}}{d_m}}-1}.$
	\item $\md_{n}=\frac{p_{n-1}}{1-p_{n-1}} \mu_{n-1}$.
	\item $p_{n}=  \frac{p_a e^{-\frac{1-a_{n}}{d_m}}+p_b e^{-\frac{1-b_{n}}{d_m}} }{1-\left(p_a e^{-\frac{1-a_{n}}{d_m}}+p_b e^{-\frac{1-b_{n}}{d_m}}\right)} \left(e^{\frac{\md_{n}}{d_m}}-1\right).$
	\end{enumerate}	

Else if $((1-a_{n-1}) > \md_{n-1}\geq(1-b_{n-1}))$ and $(b_{n-1}<1)$:
	\begin{enumerate}
	\item $a_n=a_{n-1}+\md_{n-1}, b_n=1.$
	\item $\tilde{p} =p_a \left( e^{-\frac{1-a_{n-1}-\md_{n-1}}{d_m}} - e^{-\frac{1-a_{n-1}}{d_m}}\right)+p_b \left(1- e^{-\frac{1-b_{n-1}}{d_m}}\right)$

	$\mu_{n-1}=\frac{ p_a e^{-\frac{1-a_{n}}{d}}(1+d_m -(1+\md_{n-1}+d_m)e^{-\frac{\md_{n-1}}{d_m}} ) + p_b(b_{n-1}+\md_{n-1}+d_m -(1+\md_{n-1}+d_m)e^{-\frac{1-b_{n-1}}{d}} )}{\tilde{p}}$
	\item $\md_{n}=\frac{p_{n-1}}{1-p_{n-1}} \mu_{n-1}$.
	\item $p_{n}= 1- \frac{p_a (1- e^{-\frac{1-a_{n}}{d_m}})}{1-(p_a e^{-\frac{1-a_{n-1}}{d_m}} +p_b) }.$
	\end{enumerate}	

Else if $(\md_{n-1}<(1-a_{n-1}))$ and $(a_{n-1}<1)$ and $(b_{n-1}==1)$, 
	\begin{enumerate}
	\item $a_n=a_{n-1}+\md_{n-1}.$
	\item $\mu_{n-1}=1+d_m - \frac{\md_{n-1}}{e^{\frac{\md_{n-1}}{d_m}}-1}.$
	\item $\md_{n}=\frac{p_{n-1}}{1-p_{n-1}} \mu_{n-1}$.
	\item $p_{n}=  \frac{e^{-\frac{1-a_{n}}{d_m}}}{1-e^{-\frac{1-a_{n}}{d_m}}} (e^{\frac{\md_{n}}{d_m}}-1).$
        \end{enumerate}	
Else stop.
\end{enumerate}

In Figure \ref{fig:bimodal_p}, we plot the series $\{p_n\}$ for different values of $d_m$. For this simulation, we choose $a_0=0.5, b_0=0.9, p_a=0.25, p_b=0.75.$ The mean value of the initial load is $0.25*0.5+0.75*0.9=0.8$. There is again a critical value $d_{critical}=0.02$ for the mean of the initial disturbance below which the cascading subsides, and above which all the nodes fail with probability 1. Compare this to the results of the previous section where we considered the case when all the nodes in the network had the same initial load of 0.8. In that case, $d_{critical}$ was 0.049, more than double of what it is now. This leads to an interesting question, what configuration of initial loads at the generators is best in terms of network resilience. Is it best to have all the nodes share equally the initial load?   
\begin{figure}[ht]
\centering
\includegraphics[width=0.6\columnwidth]{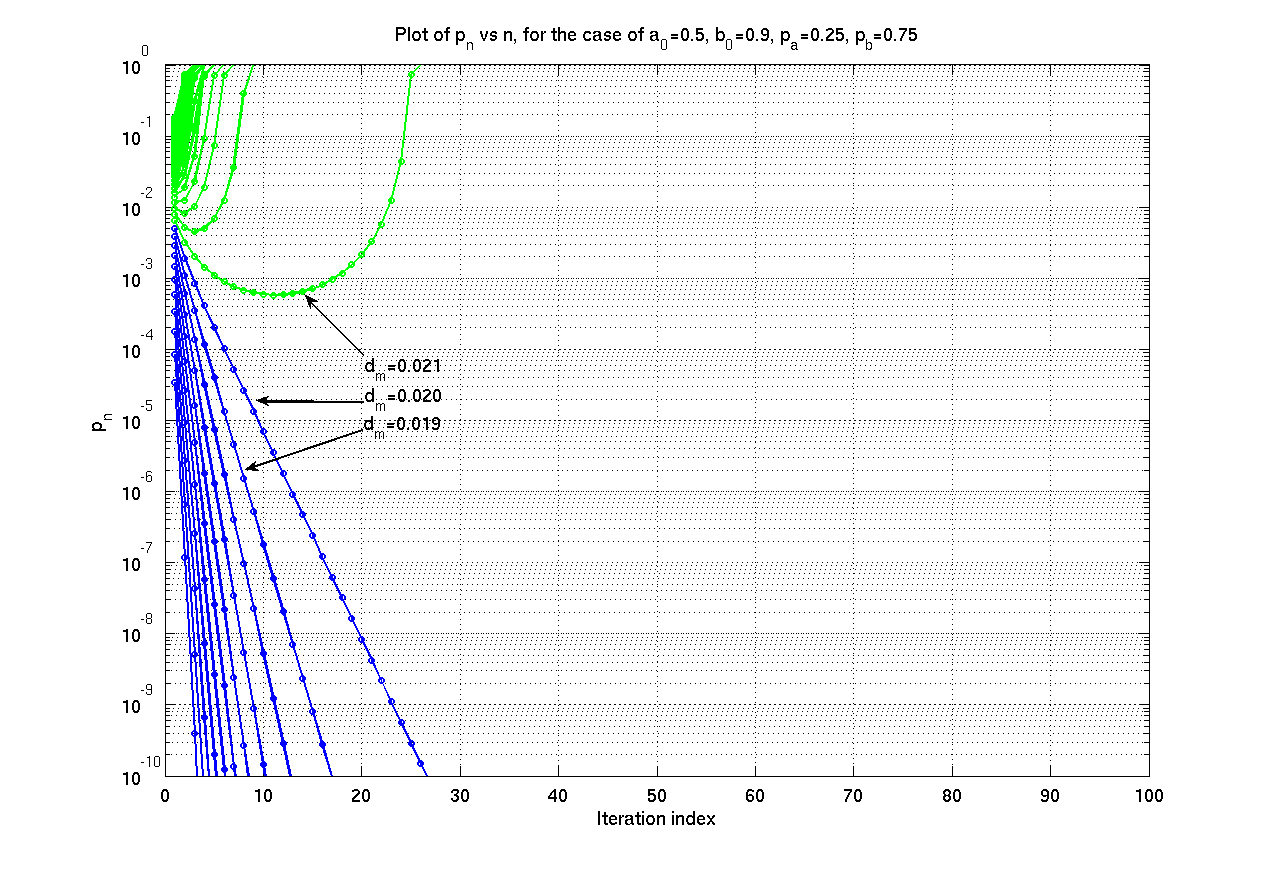}
\caption{Plot of $p_n$ vs $n$ when the initial distribution of the loads follows a bimodal distribution.}
\label{fig:bimodal_p}
\end{figure}

To verify this, we varied $a_0$ in the range $(0,0.8)$, and $b_0$ in the range $(0.8,1)$, and varied $p_a$ in a way such that we keep the mean $p_a a_0 + (1-p_a)b_0$ constant at $0.8$. For each pair of $(a_0,b_0)$, we found the critical value. In Figure \ref{fig:dcrit_a_b}, we plot $d_{critical}$. This simulation confirms our intuition that it is best to run all the generators at the same load. 

\begin{figure}[ht]
\centering
\includegraphics[width=0.6\columnwidth]{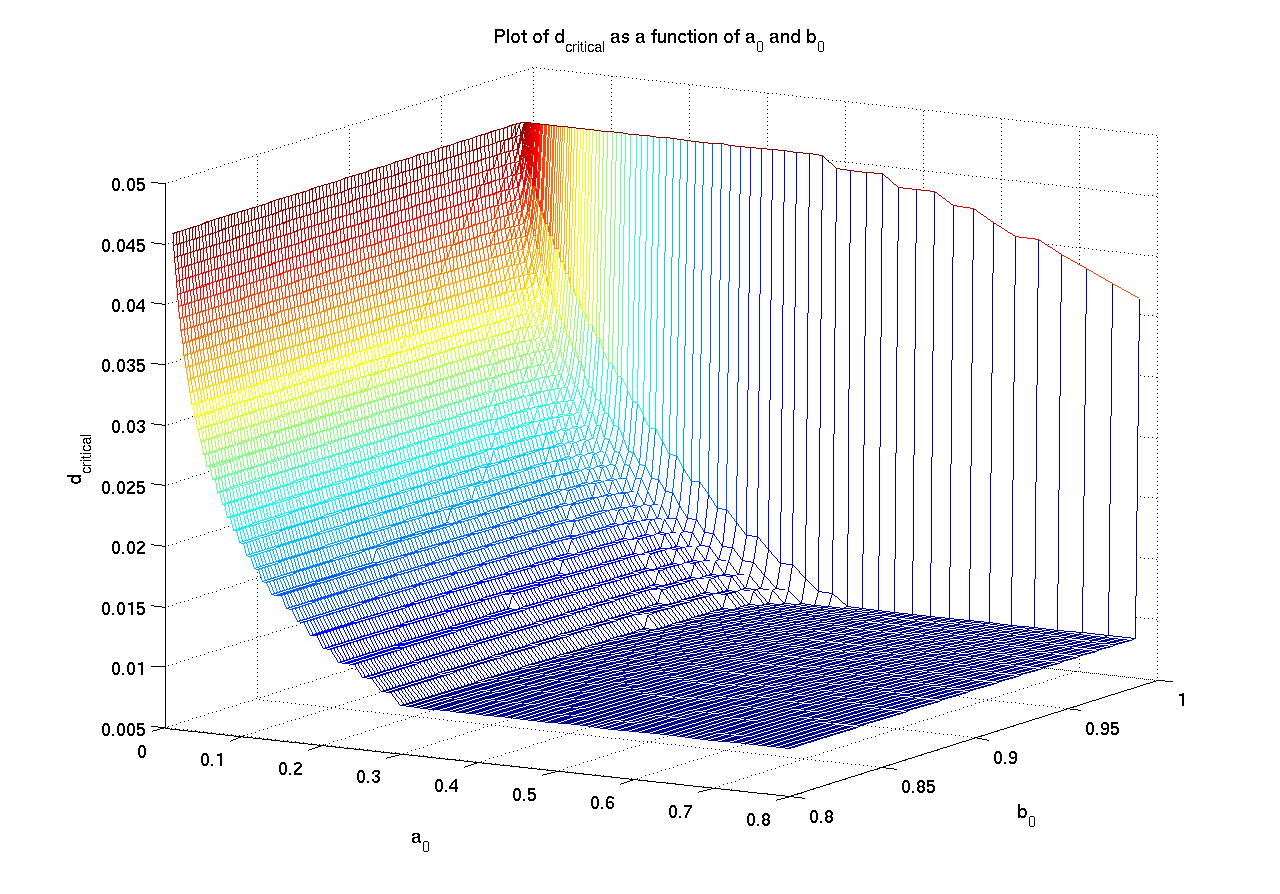}
\caption{Plot of $d_{critical}$ as a function of $a_0$ and $b_0$.}
\label{fig:dcrit_a_b}
\end{figure}

\appendix
\section{Description of the simulation}\label{app:simdesc}
We consider a network with $N$ nodes. We generate a random graph to connect these nodes. Any two nodes are connected with probability $p$, where $p\in[0,1]$ is a parameter we vary. The presence or absence of an edge is independent of the presence or absence of any other edge. Let $\textbf{A}=[a_{ij}]_{N\times N}$ be the adjacency matrix, where $a_{ij}=a_{ji}$ is 1 with probability $p$, 0 with probability 1-$p$. Define $\textbf{B}$ to be the \emph{normalized} adjacency matrix, defined as follows: $b_{ij}=a_{ij}/\sum_{i=1}^{N}a_{ij}$. The reason for defining this matrix is that the term $b_{ij}$ now represents the fraction of power transferred from node $i$ to node $j$ in the event that node $i$ fails. Note that the transpose of the normalized adjacency matrix, $B^T$, is a stochastic matrix, i.e.~all the entries of the matrix are non-negative, and the sum of elements along any row add up to 1. 

The load of node $i$ is initialized to $l_i$ distributed i.i.d.~$Unif([0,1])$. A disturbance $d_i$ is added to every node, an exponentially distributed random variable, $d_{i}$ distributed i.i.d~$Exponential(d_{mean})$. Let $\textbf{l}(0)=[l_i+d_i]_{N\times1}$ denote the load vector at time 0.

With this notation, to simulate a cascade, the following algorithm is used:
\begin{enumerate}
\item Initialize t=0.
\item Find the indices and number of nodes which fail in the current stage. Let $i_1,i_2,\ldots, i_K$ denote the indices of these nodes. 
\item Update $\textbf{A}$ by setting $a_{i_{k}i_{j}}$,$j,k\in\{1,2,\ldots,K\}$ to 0. Normalize $\textbf{A}$ to update $\textbf{B}$.
\item $\textbf{l}(t+1)=\textbf{l}(t)+\textbf{B}\times[0,\ldots,0,l_{i_1},0,\ldots,0,l_{i_2},\ldots,l_{i_K},0,\ldots,0]^T$.
\item Set $l_{i_{k}}(t+1),k\in\{1,2,\ldots,K\}$ to 0.
\item Update $\textbf{A}$ by setting $a_{i_{k}j}$ and $a_{ji_{k}}, j\in\{1,2,\ldots,N\}$, $k\in\{1,2,\ldots,K\}$ to 0. Normalize $\textbf{A}$ to update $\textbf{B}$.
\item Find the indices and number of nodes which fail in the next stage. If no new nodes fail, i.e., $\textbf{l}(t+1)<1$ (element-wise comparison), then increment $t$ and continue, else, increment $t$ and go to step 2.
\item Suppose at stage $T$, no nodes fail, compute $f=sum(\textbf{l}(T))/sum(\textbf{l}(0))$. Note that $T\leq N$. 
\item Repeat the simulation with different realizations of the random graph and different realizations of the initial load and initial disturbance. If $f_1,f_2,\ldots,f_M$ are the fractions computed from simulations $1,2,\ldots, M$, then $1-mean(f_i)$ gives the fraction of the population which goes into an outage, and is plotted in Figure \ref{fig:fraction_outage}. $mean(\mathcal{I}_{(f_i=1)})$, where $\mathcal{I}_{(\cdot)}$ is the indicator function\footnote{$\mathcal{I}_{(\mathcal{S})}$ is the indicator function evaluating to one whenever the statement $\mathcal{S}$ is true, and evaluates to 0 otherwise.}, gives the probability that that there is no outage in the system, and is plotted in Figure \ref{fig:p_outage}. The results of the individual experiment, $1-f_i$ are plotted in Figure \ref{fig:distribution}. 
\end{enumerate}

\section{Proof of lemma \ref{lem:slln}}\label{app:prooflem}
Let $X_i$ be the indicator function of whether node $i$ is dead. So, $X_1,X_2,\ldots,X_N$ are a set of i.i.d.~ Bernoulli random variables, taking the value 0 with probability $p_0$ and 1 with probability $1-p_0$. Then $N_0=\sum_{i=1}^{N}X_i$ is a Binomial random variable with parameters $N$ and $p_0$. By the strong law of large numbers, 
$$ \lim_{N\to \infty} \frac{1}{N}\sum_{i=1}^{N}X_i \stackrel{a.s.}{=} p_0 $$  
and so, $$ \lim_{N\to \infty} \frac{\frac{N_0}{N}}{1-\frac{N_0}{N}} \stackrel{a.s.}{=} \frac{p_0}{1-p_0}, $$
where $a.s.$ stands for almost sure convergence. Next we need to find the distribution of $S_{N_0}=\frac{1}{N-N_0}(L_{k_1}(0)+L_{k_2}(0)+\ldots+L_{k_{N_0}}(0))$. $N_0\stackrel{a.s.}{=}\sum_{i=1}^{N}X_i$ is a Binomial distributed random variable with mean $Np_0$ and variance $Np_0(1-p_0)$, and by central limit theorem, is equivalent in distribution to a Gaussian distribution of the same mean and variance. We can hence bound the probability of it being greater than any fixed number $N^{'}$, 
\begin{eqnarray*}
Pr(N_0 \geq N^{'}) & = & 1 - Q\left(\frac{Np_0-N^{'}}{\sqrt{Np_0(1-p_0)}}\right) \\
&>& 1- \frac{1}{\sqrt{2\pi}} \frac{\sqrt{Np_0(1-p_0)}}{Np_0-N^{'}} e^{-\frac{1}{2}\left(\frac{Np_0-N^{'}}{\sqrt{Np_0(1-p_0)}}\right)^2} \\
\lim_{N\to \infty} Pr(N_0 \geq N^{'}) &>& 1- \lim_{N \to \infty} \frac{1}{\sqrt{2\pi}} \frac{\sqrt{p_0(1-p_0)}}{\sqrt{N}(p_0-\frac{N^{'}}{N})} e^{-\frac{1}{2}\left(\frac{\sqrt{N}(p_0-\frac{N^{'}}{N})}{\sqrt{p_0(1-p_0)}}\right)^2} \\
&>& 1,
\end{eqnarray*}
where $Q(\cdot)$ is the Gaussian tail distribution. Hence, with probability 1, $N_0$ is a large number, we can make use of the weak law of large numbers to state that $$\lim_{N\to \infty} \frac{L_{k_1}(0)+L_{k_2}(0)+\ldots+L_{k_{N_0}}(0)}{N_0} \stackrel{d}{=} \mathbf{E}(L_{k_{i}}) = \mu_0, $$
where $d$ denotes convergence in distribution. Hence, $$\lim_{N\to \infty}S_{N_0}= \lim_{N\to \infty} \frac{1}{N-N_0}(L_{k_1}(0)+L_{k_2}(0)+\ldots+L_{k_{N_0}}(0)) \stackrel{d}{=} \frac{p_0}{1-p_0}\mu_0, \mathrm{ if } \mu_0<\infty$$
a constant number.
\qed
\bibliographystyle{IEEEtran}
\bibliography{citations}

\begin{thebibliography}{1}
\providecommand{\url}[1]{#1}
\csname url@samestyle\endcsname
\providecommand{\newblock}{\relax}
\providecommand{\bibinfo}[2]{#2}
\providecommand{\BIBentrySTDinterwordspacing}{\spaceskip=0pt\relax}
\providecommand{\BIBentryALTinterwordstretchfactor}{4}
\providecommand{\BIBentryALTinterwordspacing}{\spaceskip=\fontdimen2\font plus
\BIBentryALTinterwordstretchfactor\fontdimen3\font minus
  \fontdimen4\font\relax}
\providecommand{\BIBforeignlanguage}[2]{{%
\expandafter\ifx\csname l@#1\endcsname\relax
\typeout{** WARNING: IEEEtran.bst: No hyphenation pattern has been}%
\typeout{** loaded for the language `#1'. Using the pattern for}%
\typeout{** the default language instead.}%
\else
\language=\csname l@#1\endcsname
\fi
#2}}
\providecommand{\BIBdecl}{\relax}
\BIBdecl

\bibitem{Dobson04_simple}
I.~Dobson and B.~A. Carreras, ``A loading-dependent model of probabilistic
  cascading failure,'' in \emph{Probability in the Engineering and
  Informational Sciences}, 2004, pp. 15--32.

\bibitem{Dobson04_branching}
I.~Dobson, B.~Carreras, and D.~Newman, ``A branching process approximation to
  cascading load-dependent system failure,'' in \emph{System Sciences, 2004.
  Proceedings of the 37th Annual Hawaii International Conference on}, 5-8 2004,
  p. 10 pp.

\bibitem{Watts_random}
\BIBentryALTinterwordspacing
D.~J. Watts, ``A simple model of global cascades on random networks,''
  \emph{Proceedings of the National Academy of Sciences of the United States of
  America}, vol.~99, no.~9, pp. 5766--5771, 2002. [Online]. Available:
  \url{http://www.jstor.org/stable/3058573}
\BIBentrySTDinterwordspacing

\bibitem{Marchiori_physical}
P.~Crucitti, V.~Latora, and M.~Marchiori, ``Model for cascading failures in
  complex networks,'' \emph{Phys. Rev. E}, vol.~69, no.~4, p. 045104, Apr 2004.

\bibitem{Marchiori_efficiency}
\BIBentryALTinterwordspacing
V.~Latora and M.~Marchiori, ``Efficient behavior of small-world networks,''
  \emph{Physical Review Letters}, vol.~87, no.~19, pp. 198\,701+, October 2001.
  [Online]. Available: \url{http://dx.doi.org/10.1103/PhysRevLett.87.198701}
\BIBentrySTDinterwordspacing

\bibitem{Farina_optimization}
\BIBentryALTinterwordspacing
A.~Farina, A.~Graziano, F.~Mariani, and F.~Zirilli, ``Probabilistic analysis of
  failures in power transmission networks and phase transitions: Study case of
  a high-voltage power transmission network,'' \emph{Journal of Optimization
  Theory and Applications}, vol. 139, no.~1, pp. 171--199, October 2008.
  [Online]. Available: \url{http://dx.doi.org/10.1007/s10957-008-9435-x}
\BIBentrySTDinterwordspacing

\end{thebibliography}

\end{document}